\documentclass[11pt]{article}

\usepackage[utf8]{inputenc}
\usepackage[english]{babel}

\usepackage{hyperref}
\usepackage{xcolor}
\usepackage{amsmath,amssymb}
\usepackage{enumerate}
\usepackage{amsthm}
\usepackage[all,2cell]{xy}
\usepackage{mathrsfs}
\usepackage{mathtools}
\usepackage{tikz}

\usepackage{geometry}
\def\baselinestretch{1.1}
\newtheorem{thm}{Theorem}[section]
\newtheorem{dfn}[thm]{Definition}
\newtheorem{prop}[thm]{Proposition}

\newtheorem{lem}[thm]{Lemma}
\newtheorem{exmpl}[thm]{Example}
\newtheorem{obs}[thm]{Remark}

\def\beq{\begin{equation}}
\def\eeq{\end{equation}}
\def\bea{\begin{eqnarray}}
\def\eea{\end{eqnarray}}
\def\beann{\begin{eqnarray*}}
\def\eeann{\end{eqnarray*}}
\def\ben{\begin{enumerate}}
\def\een{\end{enumerate}}
\def\bit{\begin{itemize}}
\def\eit{\end{itemize}}

\def\derpar#1#2{\frac{\partial{#1}}{\partial{#2}}}

\newcommand\restr[2]{{
  \left.\kern-\nulldelimiterspace 
  #1 
  \right|_{#2} 
}}

\newcommand{\R}{\mathbb{R}}

\renewcommand{\d}{\mathrm{d}}
\newcommand{\diff}{\mathrm{d}}
\renewcommand{\L}{\mathcal{L}}
\renewcommand{\H}{\mathcal{H}}
\newcommand{\vf}{\mathfrak{X}}
\newcommand{\df}{\Omega}
\newcommand{\Cinfty}{\mathscr{C}^\infty}

\newcommand{\Tan}{\mathrm{T}}
\newcommand{\inn}{i}
\newcommand{\Lie}{\mathscr{L}}

\newcommand{\X}{\mathfrak{X}}
\newcommand{\Reeb}{\mathcal{R}}

\def\d{\mathrm{d}}

\let\ds\displaystyle

\title{\sc
A $k$-contact Lagrangian formulation\\
for nonconservative field theories}

\author{\sffamily 
$^a$Jordi Gaset, 
$^b$Xavier Gr\`acia, 
$^b$Miguel C. Mu\~noz-Lecanda,
$^b$Xavier Rivas and 
$^b$Narciso Rom\'an-Roy%
\thanks{emails: 
jordi.gaset@uab.cat,
xavier.gracia@upc.edu,
miguel.carlos.munoz@upc.edu,
xavier.rivas@upc.edu,
narciso.roman@upc.edu}
\\[1ex]
\normalsize\itshape\sffamily 
$^a$Department of Physics,
Universitat Aut\`onoma de Barcelona,
Bellaterra, Catalonia, Spain
\\[0.1ex]
\normalsize\itshape\sffamily 
$^b$Department of Mathematics,
Universitat Polit\`ecnica de Catalunya,
Barcelona, Catalonia, Spain
}

\date{\sffamily  February 23, 2020}

\begin{document}

\maketitle

\begin{abstract}

We present a geometric Lagrangian formulation for
first-order field theories with dissipation. This formulation is based on the $k$-contact geometry introduced in a previous paper, and gathers contact Lagrangian mechanics with $k$-symplectic Lagrangian field theory together.
We also study the symmetries and dissipation laws for these nonconservative theories,
and analyze some examples.
\end{abstract}

\noindent\textbf{Keywords:}
contact structure, field theory, Lagrangian system, 
dissipation, $k$-symplectic structure, $k$-contact structure.

\noindent\textbf{MSC\,2020 codes:}
70S05, 
70S10,
70G45,
53C15,
53D10,
35R01



\medskip
\setcounter{tocdepth}{2}
{
\def\baselinestretch{0.97}
\small
\def\addvspace#1{\vskip 1pt}
\parskip 0pt plus 0.1mm
\tableofcontents
}

\newpage

\section{Introduction}

In the last years the methods of differential geometry
have been used to develop an intrinsic framework
to describe dissipative or damped systems,
in particular using contact geometry
\cite{Banyaga2016,Geiges2008,Kholodenko2013}.
It has been applied to give both the
Hamiltonian and the Lagrangian descriptions of mechanical systems with dissipation
\cite{Bravetti2017,BCT-2017,CG-2019,CCM-2018,DeLeon2019,
DeLeon2016b,GGMRR-2019b,Lainz2018,LIU2018}.
Contact geometry has other physical applications, as for instance
thermodynamics, quantum mechanics, circuit theory, control theory, etc
(see \cite{Bravetti2018,CCM-2018,Goto-2016,Kholodenko2013,RMS-2017}, among others).
All of them are described by ordinary differential equations
to which some terms that account for the dissipation or damping have been added.

These geometric methods have been also used to give 
intrinsic descriptions of the Lagrangian and Hamiltonian formalisms of field theory;
in particular, those of  multisymplectic and $k$-symplectic geometry
(see, for instance, 
\cite{Carinena1991,DeLeon2015,EMR-96,GIMMSY-mm,Rey2004,Roman2009} 
and references therein).
Nevertheless, 
all these methods are developed, in general, to model systems
of variational type; that is, without dissipation or damping.

In a recent paper \cite{GGMRR-2020}
we have introduced a generalization of both 
contact geometry and $k$-symplectic geometry 
to describe field theories with dissipation,
and more specifically their Hamiltonian
(De Donder--Weyl) covariant formulation.
This new formalism is inspired by contact Hamiltonian mechanics, 
where the addition of a ``contact variable''~$s$ 
allows to describe dissipation terms;
geometrically this new variable comes from a contact form
instead of the usual symplectic form of Hamiltonian mechanics.
In the field theory case, 
if $k$ is the number of independent variables
(usually space-time variables),
we add $k$ new dependent variables $s^\alpha$ 
to introduce dissipation terms 
in the De Donder--Weyl equations.
These new variables can be obtained geometrically from the notion of 
\emph{$k$-contact structure}: 
a family of $k$ differential 1-forms $\eta^\alpha$ 
satisfying certain properties. 
Then a \emph{$k$-contact Hamiltonian system}
is a manifold endowed with a $k$-contact structure and a Hamiltonian function~$\H$.
With these elements we can state the $k$-contact Hamilton equations, 
which indeed add dissipation terms to the usual Hamiltonian field equations.
The study of their symmetries also allows to obtain some dissipation laws.
This formalism was applied to two relevant examples:
the damped vibrating string and Burgers' equation.

The aim of this paper is to extend the above study, 
developing the Lagrangian formalism of field theories with dissipation, mainly in the regular case.
For this purpose, the aforementioned $k$-contact structure 
will be used to generalize the Lagrangian formalism of the contact mechanics
presented in \cite{DeLeon2019,GGMRR-2019b}
and the Lagrangian $k$-symplectic formulation of classical field theories
\cite{DeLeon2015,Rey2004}.
In this new formalism the phase bundle
is $\oplus^k \Tan Q\times\R^k =(\Tan Q \oplus \stackrel{k}{\dots} \oplus \Tan Q)\times\R^k$.
Then, given a Lagrangian function
${\cal L}\colon\oplus^k \Tan Q\times\R^k\to\R$, 
one defines $k$ differential 1-forms $\eta^\alpha_{\cal L}$ which,
when ${\cal L}$ is {\sl regular}, constitute
a $k$-contact structure on the phase bundle.
The $k$-contact Lagrangian field equations
are then defined as the $k$-contact Hamiltonian field equations for the Lagrangian energy $E_\L$.
When written in coordinates they are the Euler--Lagrange  equations for $\L$
with some additional terms which account for dissipation.

We also study several types of symmetries for these Lagrangian field theories, as well as
their associated dissipation laws, which are characteristic of dissipative systems, 
and are the analogous to the conservation laws for conservative systems.

As examples of this formalism we study the construction of a $k$-contact Lagrangian formulation for
a class of second-order elliptic and hyperbolic partial differential equations,
and we exemplify this procedure with the equation of the
damped vibrating membrane.
In another example 
we illustrate the difference between the linear terms that appear in the equations arising from magnetic-like terms 
and those coming from a $k$-contact formulation.

The paper is organized as follows.
Section \ref{prel} is devoted to briefly review several preliminary concepts on
$k$-symplectic manifolds, $k$-contact geometry and 
$k$-contact Hamiltonian systems for field theories with dissipation.
In Section \ref{klagft} we introduce the notion of $k$-contact Lagrangian system,
and set the geometric framework for the Lagrangian formalism of field theories 
with dissipation, stating the geometric form of the 
contact Euler--Lagrange equations in several equivalent ways,
as well as the Legendre transformation and the associated canonical Hamiltonian formalism.
In Section \ref{symms} we study several types
of Lagrangian symmetries and the relations between them, as well as the corresponding dissipation laws.
Finally, some examples are given in
Section~\ref{examples}.

Throughout the paper all the manifolds and mappings are assumed to be smooth. 
Sum over crossed repeated indices is understood.

\section{Preliminaries}
\label{prel}

\subsection{$k$-tangent bundle, $k$-vector fields and geometric structures}
\label{1sect}

(See \cite{DeLeon2015,Rey2004} for more details).

Let $Q$ be a manifold and consider $\oplus^k \Tan Q = \Tan Q \oplus \stackrel{k}{\dots} \oplus \Tan Q$
(it is called the \textbf {$k$-tangent bundle} or {\sl bundle of $k^1$-velocities} of $Q$),
which is endowed with the natural projections to each direct summand and to the base manifold:
$$
\tau_\alpha\colon \oplus^k \Tan Q \rightarrow \Tan Q
\quad , \quad \tau^1_Q\colon \oplus^k \Tan Q \to Q \ .
$$
A point of $\oplus^k \Tan Q $ is ${\bf w}_q= (v_{1q},\ldots,v_{kq}) \in\oplus^k \Tan Q$,
where $(v_i)_q\in\Tan_qQ$.

A \textbf{$k$-vector field} on $Q$ is a section
${\bf X} \colon Q \longrightarrow \oplus^k \Tan Q$ of the projection~$\tau_Q^1$.
It is specified by giving 
$k$ vector fields $X_{1}, \dots, X_{k}\in\vf(Q)$, obtained as
$X_\alpha=\tau_\alpha\circ{\bf X}$; for $1\leq\alpha\leq k$, and it is denoted ${\bf X}=(X_1, \ldots, X_k)$.

Given a map $\phi\colon D\subset\R^k \rightarrow Q$,
the \textbf{first prolongation} of $\phi$ to $\oplus^k \Tan Q$
is the map $\phi' \colon D\subset\R^k \to \oplus^k \Tan Q$ defined by
$$
\phi'(t) =
\left(\phi(t),\Tan \phi \left(\frac{\partial}{\partial t^1}\Big\vert_{t}\right),
\ldots,
\Tan\phi\left(\frac{\partial}{\partial t^k}\Big\vert_{t}\right)\right)
\equiv (\phi(t);\phi_\alpha'(t)) \,,
$$
where $t = (t^1,\ldots,t^k)$ are the canonical coordinates of $\R^k$.
A map $\varphi\colon D\subset\R^k\to \oplus^k\Tan Q$
is said to be {\bf holonomic} if it is the first prolongation of a map
$\phi\colon D\subset\R^k \rightarrow Q$.

A map
$\phi \colon D\subset\R^k \rightarrow Q$ is an \textbf{integral map} of a $k$-vector field
${\bf X}=(X_{1},\dots, X_{k})$ when
\beq
\phi' = {\bf X} \circ \phi
\,.
\label{integsec}
\eeq
Equivalently,
$\ds\Tan \phi \circ \frac{\partial}{\partial t^\alpha}=X_\alpha \circ \phi$,
for every~$\alpha$.
A $k$-vector field ${\bf X}$ is \textbf{integrable} if
every point of~$Q$ is in the image of an integral map of~${\bf X}$.

In coordinates, if
$\displaystyle X_\alpha= X_\alpha^i \frac{\partial}{\partial x^i}$,
then $\phi$ is an integral map of $\mathbf{X}$ if, and only if,
it is a solution to the following system of partial differential equations:
$$
\frac{\partial \phi^i}{\partial t^\alpha} = X_\alpha^i(\phi) \ .
$$
A $k$-vector field ${\bf X}=(X_1, \ldots, X_k)$ is integrable
if, and only if, $[X_\alpha,X_\beta] = 0$, for every $\alpha,\beta$ \cite{Lee2013};
these are the necessary and sufficient conditions for the integrability of the above system of partial differential equations.

As in the case of the tangent bundle,
local coordinates $(q^i)$ in $U \subset Q$
induce natural coordinates $(q^i ,v_\alpha^i)$ in $(\tau^1_Q)^{-1}(U)\subset\oplus^k \Tan Q$,
with $1\leq i\leq n$ and $1\leq\alpha\leq k$.

Given  $\alpha$ and ${\bf w}_q\in \oplus^k \Tan Q$, there exists a natural map
$(\Lambda_q^{{\bf w}_q})^\alpha\colon\Tan_qQ\to\Tan_{{\bf w}_q}( \oplus^k \Tan Q)$,
called the {\bf $\alpha$-vertical lift} from $q$ to ${\bf w}_q$,
defined as
  $$
(\Lambda_q^{{\bf w}_q})^\alpha(u_q)= \displaystyle\frac{d}{d\lambda} (
{v_1}_q,\ldots,{v_{\alpha-1}}_q,{v_\alpha}_q+\lambda u_q,{v_{\alpha+1}}_q,
\ldots,{v_k}_q)_{\vert_{\lambda=0}} \ .
$$ 
In coordinates, if $\ds u_q = a^i \displaystyle\frac{\partial}{\partial q^i}\Big\vert_q$, we have
$\ds (\Lambda_q^{{\bf w}_q})^\alpha(u_q)= a^i \displaystyle\frac{\partial}{\partial
v^i_\alpha}\Big\vert_{w_q}$.
Observe that these $\alpha$-vertical lifts are $\tau^1_Q$-vertical vectors.
These vertical lifts extend to vector fields in a natural way; that is,
if $X\in\vf(Q)$, then its $\alpha$-vertical lift, $\Lambda^\alpha(X)\in\vf(\oplus^k \Tan Q)$,
is given by $(\Lambda^\alpha(X))_{{\bf w}_q}:=(\Lambda_q^{{\bf w}_q})^\alpha(X_q)$.

The {\bf canonical $k$-tangent structure} on $\oplus^k \Tan Q$ is the set
$(J^1,\ldots,J^k)$ of tensor fields of type $(1,1)$ in $\oplus^k \Tan Q$ defined as
$$
J^\alpha_{{\bf w}_q}:=(\Lambda_q^{{\bf w}_q})^\alpha\circ\Tan_{{\bf w}_q}\tau^1_Q \ .
$$
In natural coordinates we have
$\displaystyle J^\alpha=\frac{\partial}{\partial v^i_\alpha} \otimes \d q^i$.

The {\bf Liouville vector field}
$\Delta\in\vf(\oplus^k \Tan Q)$  is the infinitesimal generator of the flow
$\psi\colon\R\times \oplus^k \Tan Q\longrightarrow \oplus^k \Tan Q$, given by
$\psi(t;v_{1q},\ldots,v_{kq}) = (e^t v_{1q},\ldots,e^t v_{kq})$.
Observe that $\Delta=\Delta_1+\ldots+\Delta_k$, where each $\Delta_\alpha\in\vf(\oplus^k \Tan Q)$ is the infinitesimal generator of the flow $\psi^\alpha\colon \R \times \oplus^k \Tan Q \longrightarrow \oplus^k \Tan Q$
$$
\psi^\alpha(s;v_{1q},\ldots,v_{kq})=(v_{1q},\ldots,v_{(\alpha-1)q}, e^t v_{\alpha q},v_{(\alpha+1)q},\ldots,v_{kq}) \ .
$$
In coordinates, 
$\displaystyle\Delta = v^i_\alpha\derpar{}{v_\alpha^i}$.

Given a map $\Phi\colon M\to N$, there exists a natural extension
$\oplus^k\Tan\Phi\colon\oplus^k\Tan M\to\oplus^k\Tan N$, defined by
$$
\oplus^k\Tan\Phi(v_{1q},\ldots,v_{kq}):=(\Tan_q\Phi(v_{1q}),\ldots,\Tan_q\Phi(v_{kq})) \ .
$$
By definition, a $k$-vector field $\mathbf{\Gamma}=(\Gamma_1,\ldots,\Gamma_k)$ in $\oplus^k \Tan Q$
is a section of the projection
$$
\tau^1_{\oplus^k \Tan Q}\colon\Tan(\oplus^k \Tan Q)\oplus \stackrel{k}{\dots}\oplus\Tan(\oplus^k \Tan Q) \rightarrow\oplus^k \Tan Q \ .
$$
Then, we say that $\mathbf{\Gamma}$ is a
\textbf{second order partial differential equation ({\sc sopde})} 
if it is also a section of the projection
$$
\oplus^k\Tan\tau^1_Q \colon\Tan(\oplus^k \Tan Q)\oplus \stackrel{k}{\dots}\oplus\Tan(\oplus^k \Tan Q) \rightarrow\oplus^k \Tan Q \ ;
$$
that is,
$\oplus^k\Tan\tau^1_Q\circ\mathbf{\Gamma}={\rm Id}_{\oplus^k \Tan Q}=
\tau^1_{\oplus^k \Tan Q}\circ\mathbf{\Gamma}$.
Notice that a $k$-vector field
$\mathbf{\Gamma}$ in $\oplus^k \Tan Q$ is a {\sc sopde}
if, and only if,
$J^\alpha(\Gamma_\alpha)=\Delta$.

In addition, an integrable $k$-vector field $\mathbf{\Gamma}=(\Gamma_1,\ldots,\Gamma_k)$ in
$\oplus^k \Tan Q$ is a {\sc sopde}
if, and only if, its integrable maps are holonomic.

In natural coordinates, the expression of the components of a {\sc sopde} is
$\displaystyle \Gamma_\alpha=
v^i_\alpha\frac{\partial} {\partial q^i}+
\Gamma_{\alpha\beta}^i \frac{\partial} {\partial v^i_\beta}$.
Then, if $\psi\colon\R^k \to\oplus^k \Tan Q$, locally given by
$\psi(t)=(\psi^i(t),\psi^i_\beta(t))$, is an integral map of
an integrable {\sc sopde},
from \eqref{integsec} we have that
$$
\frac{\partial\psi^i} {\partial t^\alpha}\Big\vert_t=\psi^i_\alpha(t) 
\quad , \quad
\frac{\partial\psi^i_\beta} {\partial t^\alpha}\Big\vert_t=\Gamma_{\alpha\beta}^i(\psi(t))
\, .
$$
Furthermore, $\psi=\phi'$, where $\phi'$ is the first prolongation of the map
$\phi=\tau\circ\psi \colon
\R^k \stackrel{\psi}{\to}\oplus^k \Tan Q \stackrel{\tau}{\to} Q$,
and hence $\phi$ is a solution to the system of second order partial
differential equations
\beq
\label{nn1}
\frac{\partial^2 \phi^i}{\partial t^\alpha\partial t^\beta}(t)=
\Gamma_{\alpha\beta}^i\left(\phi^i(t),\frac{\partial\phi^i}{\partial
t^\gamma}(t)\right) \ .
\eeq
Observe that, from \eqref{nn1} we obtain that, if $\mathbf{\Gamma}$ is an
integrable {\sc sopde}, then 
$\Gamma_{\alpha\beta}^i=\Gamma_{\beta\alpha}^i$.

\subsection{\texorpdfstring{$k$}--symplectic manifolds}

(See 
\cite{Awane1992,DeLeon1988,DeLeon1988a,DeLeon2015,Rey2004} 
for more details.)

Let $M$ be a manifold of dimension $N=n+kn$.
A \textbf{$k$-symplectic structure} on~$M$
is a family $(\omega^1,\ldots,\omega^k;V)$,
where $\omega^\alpha$ ($\alpha=1,\ldots,k$) are closed $2$-forms,
and $V$ is an integrable $nk$-dimensional tangent distribution on~$M$
such that
$$
(i) \ \omega^\alpha \vert_{V\times V} =0 \ \hbox{\rm (for every $\alpha$)}
\:, \quad
(ii)\bigcap_{\alpha=1}^{k} \ker\omega^\alpha = \{0\} 
\:.
$$
Then $(M,\omega^\alpha,V)$ is called a \textbf{$k$-symplectic manifold}.

For every point of $M$ there exist a neighbourhood $U$
and local coordinates
$(q^i , p^\alpha_i)$ ($1\leq i\leq n$, $1\leq \alpha\leq k$)
such that, on~$U$,
$$
\omega^\alpha=  \d q^i\wedge \d p^\alpha_i 
\:,\quad
V =
\left\langle \frac{\partial}{\partial p^1_i}, \dots,
\frac{\partial}{\partial p^k_i} \right\rangle 
\,.
$$
These are the so-called \emph{Darboux} or \emph{canonical coordinates}
of the $k$-symplectic manifold \cite{Awane1992}.

The canonical model for $k$-symplectic manifolds is
$\oplus^k \Tan^*Q= \Tan^*Q\oplus \stackrel{k}{\dots} \oplus \Tan^*Q$,
with natural projections
$$
\pi^\alpha\colon \oplus^k \Tan^*Q \rightarrow\Tan^*Q
\:, \quad 
\pi^1_Q \colon \oplus^k \Tan^*Q \to Q 
\:.
$$
As in the case of the cotangent bundle,
local coordinates $(q^i)$ in $U \subset Q$ induce natural coordinates $(q^i ,p^\alpha_i)$ in
$(\pi^1_Q)^{-1}(U)$.
If $\theta$ and $\omega=-\d\theta$ are the canonical forms of $\Tan^*Q$,
then $\oplus^k \Tan^*Q$ is endowed with the canonical forms
\beq
\theta^\alpha=(\pi^\alpha)^*\theta 
\,,\quad
\omega^\alpha=(\pi^\alpha)^*\omega=-(\pi^\alpha)^*\d\theta=-\d\theta^\alpha 
,
\label{cartforms}
\eeq
and in natural coordinates we have that 
$\theta^\alpha = p^\alpha_i\d q^i$ 
and
$\omega^\alpha=\d q^i\wedge\d p^\alpha_i$.
Thus, the triple  
$(\oplus^k \Tan^*Q,\omega^\alpha, V)$,
where $V=\ker \Tan \pi^1_Q$,
is a $k$-symplectic manifold,
and the natural coordinates in $\oplus^k \Tan^*Q$ are Darboux coordinates.

\subsection{\texorpdfstring{$k$}--contact structures}
\label{kcontact}

The definition of $k$-contact structure has been recently introduced in
\cite{GGMRR-2020},
where the reader can find more details.

Remember that,
if $M$ is a smoooth manifold of dimension~$m$,
a (generalized) distribution on~$M$ 
is a subset $D \subset \Tan M$
such that, for every $x \in M$,
$D_x \subset \Tan_xM$ is a vector subspace.
The distribution $D$ is smooth when it can be locally spanned
by a family of smooth vector fields, and
is regular when it is smooth
and has locally constant rank.
A codistribution on $M$ is a subset $C \subset \Tan^*M$ with similar properties.
The annihilator $D^\circ$ of a distribution~$D$  is a codistribution.

A (smooth) differential 1-form $\eta \in \Omega^1(M)$
generates a smooth codistribution that we denote by
$\langle \eta \rangle \subset \Tan^*M$;
it has rank~1 at every point where $\eta$ does not vanish.
Its annihilator is a distribution 
$\langle \eta \rangle^\circ \subset \Tan M$;
it can be described also as the kernel
of the vector bundle morphism
$\widehat \eta \colon \Tan M \to M \times \R$
defined by~$\eta$.
This distribution has corank~1 
at every point where $\eta$ does not vanish.

Now, given $k$ differential 1-forms
$\eta^1, \ldots, \eta^k \in \Omega^1(M)$, let:
\beann
\mathcal{C}^{\mathrm{C}} &=&
\langle \eta^1, \ldots, \eta^k \rangle \subset
\Tan^*M \ ,\\
\mathcal{D}^{\mathrm{C}} &=&
\left( \mathcal{C}^{\mathrm{C}} \right)^\circ =
\ker \widehat{\eta^1} \cap \ldots \cap \ker \widehat{\eta^k} \subset
\Tan M \ , \\
\mathcal{D}^{\mathrm{R}} &=&
\ker \widehat{\d \eta^1} \cap \ldots \cap \ker \widehat{\d \eta^k} \subset
\Tan M \ , \\
\mathcal{C}^{\mathrm{R}} &=& 
\left( \mathcal{D}^{\mathrm{R}} \right)^\circ\subset
\Tan^*M \ .
\eeann

\begin{dfn}
\label{kconman}
A \textbf{$k$-contact structure} on $M$ is a family of $k$ differential 1-forms 
$\eta^\alpha \in \Omega^1(M)$ such that, with the preceding notations,
\begin{enumerate}[(i)]
\item 
$\mathcal{D}^{\mathrm{C}} \subset \Tan M$
is a regular distribution of corank~$k$;
or, what is equivalent,
$\eta^1 \wedge \ldots \wedge \eta^k \neq 0$, at every point.
\item 
$\mathcal{D}^{\mathrm{R}} \subset \Tan M$
is a regular distribution of rank~$k$.
\item
$\mathcal{D}^{\mathrm{C}} \cap \mathcal{D}^{\mathrm{R}} = \{0\}$  or, what is equivalent,
$\displaystyle
\bigcap_{\alpha=1}^{k} \left(
\ker \widehat{\eta^\alpha} \cap \ker \widehat{\d \eta^\alpha}\right) =\{0\}$.
\end{enumerate}
We call
$\mathcal{C}^{\mathrm{C}}$
the \textbf{contact codistribution};
$\mathcal{D}^{\mathrm{C}}$ 
the \textbf{contact distribution};
$\mathcal{D}^{\mathrm{R}}$ 
the \textbf{Reeb distribution};
and
$\mathcal{C}^{\mathrm{R}}$
the \textbf{Reeb codistribution}.
\\
A \textbf{$k$-contact manifold} is a manifold endowed with a $k$-contact structure.
\end{dfn}

\begin{obs}\rm
If conditions (i) and (ii) hold, then (iii) is equivalent to 
$$
{\it (iii\,')}\ 
\Tan M = \mathcal{D}^{\mathrm{C}} \oplus \mathcal{D}^{\mathrm{R}}  .
$$
\end{obs}
For $k=1$ we recover the definition of contact structure.

\begin{thm}
\label{reebvf}
Let  $(M,\eta^\alpha)$ be a $k$--contact manifold.
\ben
\item
The Reeb distribution 
$\mathcal{D}^{\mathrm{R}}$ 
is involutive,
and therefore integrable.
\item
There exist $k$ vector fields 
$R_\alpha \in \mathfrak{X}(M)$,
the \emph{Reeb vector fields},
uniquely defined by the relations
\beq
\label{reebcontact}
\inn({\Reeb_\beta}) \eta^\alpha = \delta^\alpha_{\,\beta}
\,,\quad 
\inn({\Reeb_\beta}) \d\eta^\alpha = 0
\,.
\eeq
\item
The Reeb vector fields commute,
$\displaystyle [\Reeb_\alpha,\Reeb_\beta] = 0$,
and they generate $\mathcal{D}^{\mathrm{R}}$.
\een
\end{thm}

There are coordinates
$(x^I;s^\alpha)$
such that
$$
    \Reeb_\alpha = \frac{\partial}{\partial s^\alpha}
    \,,\quad
    \eta^\alpha = \d s^\alpha - f_I^\alpha(x) \,\d x^I
    \,,
$$
where $f_I^\alpha(x)$ are functions depending only on the~$x^I$,
which are called $\textbf{adapted coordinates}$ (to the $k$-contact structure).

\begin{exmpl}
\label{example-canmodel}
\rm
Given $k \geq 1$,
the manifold
$(\oplus^k \Tan^*Q) \times \R^k$
has a canonical $k$-contact structure defined by the 1-forms
$$
\eta^\alpha = \d s^\alpha - \theta^\alpha
\,,
$$
where
$s^\alpha$
is the $\alpha$-th cartesian coordinate of $\R^k$,
and 
$\theta^\alpha$
is the pull-back of the canonical 1-form of $\Tan^*Q$
with respect to the projection
$(\oplus^k \Tan^*Q) \times \R^k\to \Tan^*Q$
to the $\alpha$-th direct summand.
Using coordinates $q^i$ on~$Q$ and natural coordinates
$(q^i,p_i^\alpha)$ on each $\Tan^*Q$,
their local expressions are
$$
\eta^\alpha = \d s^\alpha - p^\alpha_i \,\d q^i
\,,
$$
from which 
$\d \eta^\alpha = \d q^i \wedge \d p_i^\alpha$,
and the Reeb vector fields are
$$
\Reeb_\alpha = \frac{\partial}{\partial s^\alpha}
\,.
$$
\end{exmpl}

The following result ensures the existence of canonical coordinates
for a particular kind of $k$-contact manifolds:

\begin{thm}[$k$-contact Darboux theorem]
\label{Darboux k-contact}
Let  $(M,\eta^\alpha)$ be a $k$--contact manifold of dimension $n+kn+k$
such that there exists an integrable subdistribution ${\cal V}$ of ${\cal D}^{\rm C}$
with ${\rm rank}\,{\cal V}=nk$.
Around every point of $M$, there exists a local chart of coordinates 
$(U;q^i,p^\alpha_i,s^\alpha)$, $1\leq\alpha\leq k \ ,\ 1\leq i \leq n$,
 such that
$$
\restr{\eta^\alpha}{U}=\d s^\alpha-p_i^\alpha\,\d q^i \;.
$$
In these coordinates,
$$
{\cal D}^{\rm R}\vert_{U}=\left\langle\Reeb_\alpha=\frac{\partial}{\partial s^\alpha}\right\rangle 
\quad , \quad
{\cal V}\vert_{U}=\left\langle\frac{\partial}{\partial p_i^\alpha}\right\rangle
\ .
$$
These are the so-called {\rm canonical} or {\rm Darboux coordinates} of the $k$-contact manifold.
\end{thm}

This theorem allows us to consider the manifold presented in the example \ref{example-canmodel} as the canonical model
for these kinds of $k$-contact manifolds.

\subsection{\texorpdfstring{$k$}--contact Hamiltonian systems}

Together with $k$-contact structures,
$k$-contact Hamiltonian systems have also been defined in
\cite{GGMRR-2020}.

A \textbf{$k$-contact Hamiltonian system} is a family $(M,\eta^\alpha,\H)$,
where $(M,\eta^\alpha)$ is a $k$-contact manifold,
and $\H\in \Cinfty(M)$ is called a \textbf{Hamiltonian function}.
The \textbf{$k$-contact Hamilton--de Donder--Weyl equations} for a map
$\psi\colon D\subset\R^k\to M$ are
\begin{equation}
\begin{cases}
i(\psi'_\alpha)\d\eta^\alpha = \big(\d\H - (\Lie_{\Reeb_\alpha}\H)\eta^\alpha\big)\circ\psi \ ,\\
i(\psi'_\alpha)\eta^\alpha = - \H\circ\psi \ .
\end{cases}
\label{hec}
\end{equation}

The \textbf{$k$-contact Hamilton--de Donder--Weyl equations} 
for a $k$-vector field ${\bf X}=(X_1,\dots,X_k)$ in $M$ are
\beq \begin{cases}
    \inn({X_\alpha})\d\eta^\alpha=\d\H-(\Lie_{\Reeb_\alpha}\H)\eta^\alpha \ ,\\
    \inn(X_\alpha)\eta^\alpha=-\H \ .
    \end{cases}
    \label{fieldcontact}
\eeq
Their solutions are called
\textbf{Hamiltonian $k$-vector fields}.
These equations are equivalent to
\beq
\label{eq:kcontact3}
\begin{cases}
\Lie_{X_\alpha}\eta^\alpha = - (\Lie_{\Reeb_\alpha}\H)\eta^\alpha \ , \\
\inn(X_\alpha)\eta^\alpha = - \H \ .
\end{cases}
\eeq
Solutions to these equations always exist, although
they are neither unique, nor necessarily integrable.

If ${\bf X}$ is an {\sl integrable} $k$-vector field in $M$, then
every integral map $\psi \colon D\subset\R^k \to M$ of ${\bf X}$
satisfies the $k$-contact equation \eqref{hec}
if, and only if,
${\bf X}$ is a solution to \eqref{fieldcontact}.
Notice, however, that 
equations \eqref{hec} and \eqref{fieldcontact}
are not, in general, fully equivalent, since
a solution to \eqref{hec} may not be
an integral map of some integrable $k$-vector field in~$M$ solution to \eqref{fieldcontact}.

An alternative,
partially equivalent, expression for the Hamilton--De Donder--Weyl equations,
which does not use
the Reeb vector fields $\Reeb_\alpha$, can be given as follows.
Consider the 2-forms $\Omega^\alpha= -\H\,\d\eta^\alpha+\d\H\wedge\eta^\alpha$.
On the open set ${\cal O}=\{ p\in M\ ;\ \H(p)\not=0\}$,
if a $k$-vector field ${\bf X}=(X_\alpha)$ satisfies
\begin{equation}
\label{hamilton-eqs-no-reeb}
    \begin{cases}
        i(X_\alpha)\Omega^\alpha = 0\,,\\
        i(X_\alpha)\eta^\alpha = -\H\,,
    \end{cases}
\end{equation}
then ${\bf X}$ is a solution of the Hamilton--De Donder--Weyl equations \eqref{fieldcontact}).
Any integral map $\psi$ of such a $k$-vector field is a solution to
\begin{equation}
 \begin{cases}
i(\psi'_\alpha)\Omega^\alpha = 0 
\:,\\
i(\psi'_\alpha)\eta^\alpha = - \H\circ\psi 
\:.
\end{cases}
\label{hec2}
\end{equation}

\begin{obs}
{\rm If the family $(M,\eta^\alpha)$
does not hold some of the conditions of Definition \eqref{kconman},
then $(M,\eta^\alpha)$ is called a {\sl $k$-precontact manifold}
and $(M,\eta^\alpha,\H)$ is said to be a \textsl{$k$-precontact Hamiltonian system}.
In this case, the Reeb vector fields are not uniquely defined.
However, as it happens in other similar situations
(precosymplectic mechanics, $k$-precosymplectic field theories or precontact mechanics) \cite{DeLeon2019,GRR-2019},
it could be proved that equations \eqref{hec} and \eqref{fieldcontact}
does not depend on the used Reeb vector fields and, thus,
the equations are still valid.
}
\end{obs}

In canonical coordinates, if $\psi=(q^i(t^\beta),p^\alpha_i(t^\beta),s^\alpha(t^\beta))$, 
then $\displaystyle\psi'_\alpha=
\Big(q^i,p^\alpha_i,s^\alpha,\derpar{q^i}{t^\beta},\derpar{p^\alpha_i}{t^\beta},\derpar{s^\alpha}{t^\beta}\Big)$,
and these equations read
\beq
    \begin{cases}
        \displaystyle \frac{\partial q^i}{\partial t^\alpha} = \frac{\partial\H}{\partial p^\alpha_i}\circ\psi \ ,\\[15pt]
        \displaystyle \frac{\partial p^\alpha_i}{\partial t^\alpha} = 
        -\left(\frac{\partial\H}{\partial q^i}+ p_i^\alpha\frac{\partial\H}{\partial s^\alpha}\right)\circ\psi \ ,\\[15pt]
        \displaystyle \frac{\partial s^\alpha}{\partial t^\alpha} = 
        \left(p_i^\alpha\frac{\partial\H}{\partial p^\alpha_i}-\H\right)\circ\psi\ ,
    \end{cases}
    \label{coor1}
\eeq
If ${\bf X}=(X_\alpha)$ is a $k$-vector field solution to \eqref{hamilton-eqs-no-reeb}
and in canonical coordinates we have that
$\displaystyle
X_\alpha= X_\alpha^\beta\derpar{}{s^\beta}+ X_\alpha^i\frac{\partial}{\partial q^i}+
X_{\alpha i}^\beta\frac{\partial}{\partial p_i^\beta}$, 
then
\beq
\begin{cases}
        \displaystyle X_\alpha^i= \frac{\partial\H}{\partial p^\alpha_i} \ ,\\[15pt]
        \displaystyle X_{\alpha i}^\alpha = 
        -\left(\frac{\partial\H}{\partial q^i}+ p_i^\alpha\frac{\partial\H}{\partial s^\alpha}\right) \ ,\\[15pt]
        \displaystyle X_\alpha^\alpha = 
        p_i^\alpha\frac{\partial\H}{\partial p^\alpha_i}-\H\ ,
    \end{cases}
    \label{coor2}
\eeq
\section{\texorpdfstring{$k$}--contact  Lagrangian field theory}
\label{klagft}

\subsection{\texorpdfstring{$k$}--contact Lagrangian systems}

Using the geometric framework introduced in Section \ref{1sect},
we are ready to deal with 
Lagrangian systems with dissipation in field theories.
First we need to enlarge the bundle in order to include the dissipation variables.
Then, consider the bundle $\oplus^k \Tan Q\times\R^k$ with canonical projections
$$
\bar\tau_1\colon \oplus^k \Tan Q\times\R^k\to\oplus^k \Tan Q
 \quad , \quad
\bar\tau^k\colon \oplus^k \Tan Q\times\R^k\to\Tan Q
\quad , \quad
s^\alpha\colon \oplus^k \Tan Q\times\R^k\to\R \ .
$$
Natural coordinates in $\oplus^k \Tan Q\times\R^k$ are $(q^i,v^i_\alpha,s^\alpha)$.

As $\oplus^k \Tan Q\times\R^k\to\oplus^k \Tan Q$ is a trivial bundle, 
the canonical structures in $\oplus^k \Tan Q$ (the 
canonical $k$-tangent structure and the
Liouville vector field described above)
can be extended to $\oplus^k \Tan Q\times\R^k$ in a natural way,
and are denoted with the same notation $(J^\alpha)$ and $\Delta$.
Then, using these structures, we can extend also the concept of {\sc sode} $k$-vector fields to $\oplus^k \Tan Q\times\R^k$ as follows:
			
\begin{dfn}
    A $k$-vector field $\mbox{\boldmath $\Gamma$}=(\Gamma_\alpha)$ in $\oplus^k \Tan Q\times\R^k$ is a 
    \textbf{second order partial differential equation} (\textsc{sopde}) if
    $J^\alpha(\Gamma_\alpha)=\Delta$.
\end{dfn}
            
The local expression of a {\sc sopde} is
\begin{equation}
\label{localsode2}
\Gamma_\alpha= 
v^i_\alpha\frac{\displaystyle\partial} {\displaystyle
\partial q^i}+\Gamma_{\alpha\beta}^i\frac{\displaystyle\partial}{\displaystyle \partial v^i_\beta}+g^\beta_\alpha\,\frac{\partial}{\partial s^\beta}\ .
\end{equation}

\begin{dfn}
\label{de652}
Let $\psi\colon\R^k\rightarrow Q\times\R^k$ be a section of the projection
$Q\times\R^k\to\R^k$;
with $\psi=(\phi,s^\alpha)$, where
$\phi\colon\R^k\to Q$.
The \textbf{first prolongation} 
 of $\psi$ to $\oplus^k \Tan Q\times\R^k$ is the map
$\mbox{\boldmath $\sigma$}\colon\R^k\to\oplus^k \Tan Q\times\R^k$
given by
$\mbox{\boldmath $\sigma$}=(\phi',s^\alpha)$.
The map $\mbox{\boldmath $\sigma$}$ is said to be \textbf{holonomic}.
\end{dfn}

The following property is a straightforward consequence of
the above definitions and the results about {\sc sopdes} in the bundle
$\oplus^k \Tan Q$ given in Section \ref{1sect}:

\begin{prop}\label{lem0}
A $k$-vector field $\mbox{\boldmath $\Gamma$}$ in $\oplus^k \Tan Q\times\R^k$ is a {\sc sopde}
if, and only if, its integral maps are holonomic.
\end{prop}

Now we can state the Lagrangian formalism of field theories with dissipation.

\begin{dfn}
    \label{lagrangean}
    A \textbf{Lagrangian function} 
    is a function $\L\in\Cinfty(\oplus^k \Tan Q\times\R^k)$.
    
    \noindent The \textbf{Lagrangian energy}
    associated with $\L$ is the function $E_\L:=\Delta(\L)-\L\in\Cinfty(\oplus^k \Tan Q\times\R^k)$.
    
    \noindent The \textbf{Cartan forms}
    associated with $\L$ are
    \begin{equation*}
\theta_\L^\alpha={}^t(J^\alpha)\circ\d\L \in\df^1(\oplus^k \Tan Q\times\R^k)
\quad , \quad
\omega_\L^\alpha=-\d\theta_\L^\alpha\in\df^2(\oplus^k \Tan Q\times\R^k)\ .
    \end{equation*}
    Finally, we can define the forms
    $$
    \eta_\L^\alpha=\d s^\alpha-\theta_\L^\alpha\in\Omega^1(\oplus^k \Tan Q\times\R^k) \quad ,\quad \d\eta_\L^\alpha=\omega_\L^\alpha\in\Omega^2(\oplus^k \Tan Q\times\R^k) \ .
    $$
The couple $(\oplus^k \Tan Q\times\R^k,\L)$ is said to be a \textbf{$k$-contact Lagrangian system}.
\end{dfn}

In natural coordinates $(q^i,v^i_\alpha,s^\alpha)$ of $\oplus^k \Tan Q\times\R^k$,
the local expressions of these elements are
$$
E_L=v^i_\alpha\frac{\partial\L}{\partial v^i_\alpha}-\L
\quad ,\quad
\eta_\L^\alpha=\d s^\alpha-\frac{\partial\L}{\partial v^i_\alpha}\d q^i\ .
$$

Before introducing the Legendre map, remember that, given a bundle map $f \colon E \to F$
between two vector bundles over a manifold~$B$,
the fibre derivative of~$f$ is the map
$\mathcal{F}f \colon  E \to \mathrm{Hom}(E,F) \approx F \otimes E^*$
obtained by restricting $f$ to the fibres, $f_b \colon E_b \to F_b$,
and computing the usual derivative of a map between two vector spaces:
$\mathcal{F}f(e_b) = \mathrm{D} f_b(e_b)$.
This applies in particular when the second vector bundle is trivial of rank~1,
that is, for a function $f \colon E \to \R$;
then $\mathcal{F}f \colon E \to E^*$.
This map also has a fibre derivative $\mathcal{F}^2 f \colon E \to E^* \otimes E^*$,
which is usually called the fibre Hessian of~$f$. For every $e_b \in E$,
$\mathcal{F}^2 f(e_b)$ can be considered as a symmetric bilinear form on $E_b$.
It is easy to check that $\mathcal{F}f$ is a local diffeomorphism 
at a point $e \in E$ if, and only if, the Hessian $\mathcal{F}^2f(e)$ is non-degenerate.
(See \cite{Gracia2000} for details).

\begin{dfn}
The \textbf{Legendre map}  associated with a Lagrangian $\L\in\Cinfty(\oplus^k \Tan Q\times\R^k)$
is the fibre derivative of~$\L$, considered as a function on the vector bundle $\oplus^k \Tan Q\times\R^k \to Q \times \R^k$;
that is, the map
${\cal F}\L\colon\oplus^k \Tan Q\times\R^k\to\oplus^k \Tan^*Q\times\R^k$
given by
$$
{\cal F}\L ({v_1}_q,\ldots , {v_k}_q;s^\alpha) =
 \left({\cal F}\L(\cdot,s^\alpha) ({v_1}_q,\ldots , {v_k}_q),s^\alpha \right)
\ ; \quad
({v_1}_q,\ldots , {v_k}_q)\in \oplus^k \Tan Q\ ,
$$
where $\L(\cdot,s^\alpha)$ denotes the Lagrangian with $s^\alpha$ freezed.
\end{dfn}

This map is locally given by
$\displaystyle{\cal F}\L(q^i,v^i_\alpha,s^\alpha)=\Big(q^i,\frac{\partial \L}{\partial v^i_\alpha},s^\alpha\Big)$.

\begin{obs}\rm
The Cartan forms can also be defined as
$$
\theta_\L^\alpha={\cal FL}^{\;*}\theta^\alpha
\,,\quad
\omega_\L^\alpha={\cal FL}^{\;*}\omega^\alpha
\, ,
$$
where $\theta^\alpha$ and $\omega^\alpha$ are given in \eqref{cartforms}.
\end{obs}

\begin{prop}
\label{Prop-regLag}
For a Lagrangian function $\L$ the following conditions are equivalent:
\begin{enumerate}
\item
The Legendre map
${\cal FL}$ is a local diffeomorphism.
\item
The fibre Hessian
$
{\cal F}^2\L \colon
\oplus^k\Tan Q\times\R^k \longrightarrow (\oplus^k\Tan^*Q\times\R^k)\otimes (\oplus^k\Tan^*Q\times\R^k)
$
of~$\L$ is everywhere nondegenerate.
(The tensor product is of vector bundles over $Q \times \R^k$.)
\item
$(\oplus^k\Tan Q\times\R^k,\eta_\L^\alpha)$ is a $k$-contact manifold.
\end{enumerate}
\end{prop}
\begin{proof}
The proof can be easily done using natural coordinates, bearing in mind that
\beann
\displaystyle F\L(q^i,v^i_\alpha,s^\alpha)=\Big(q^i,\frac{\partial \L}{\partial v^i_\alpha},s^\alpha\Big) &,&
\\
{\cal F}^2 \L(q^i,v^i_\alpha,s^\alpha) = (q^i,W_{ij}^{\alpha\beta},s^\alpha)
&,&
\hbox{with}\ 
W_{ij}^{\alpha\beta}= 
\left( \frac{\partial^2\L}{\partial v^i_\alpha\partial v^j_\beta}\right)
\,.
\eeann
Then the conditions in the proposition 
mean that the matrix 
$W= (W_{ij}^{\alpha\beta})$ 
is everywhere nonsingular.
\end{proof}

\begin{dfn}
A Lagrangian function $\L$ is said to be \textbf{regular} if the equivalent
conditions in Proposition \ref{Prop-regLag} hold.
Otherwise $\L$ is called a \textbf{singular} Lagrangian.
In particular, 
$\L$ is said to be \textbf{hyperregular} 
if ${\cal FL}$ is a global diffeomorphism.
\end{dfn}

Given a regular $k$-contact Lagrangian system $(\oplus^k\Tan Q\times\R^k,\L)$,
from \eqref{reebcontact} we have that
the \textbf{Reeb vector fields} $(\Reeb_\L)_\alpha\in\X(\oplus^k \Tan Q\times\R^k)$ 
for this system are the unique solution to
$$
\inn((\Reeb_\L)_\alpha)\diff\eta^\beta_\L=0\quad ,\quad
\inn((\Reeb_\L)_\alpha)\eta^\beta_\L=\delta_\alpha^\beta \ .
$$
If $\L$ is regular, then there exists the inverse 
$W^{ij}_{\alpha\beta}$ of the Hessian matrix,
namely $\displaystyle W^{ij}_{\alpha\beta}\frac{\partial^2\L}{\partial v^j_\beta \partial v^k_\gamma}=\delta^i_k\delta^\gamma_\alpha$,
and then a simple calculation in coordinates leads to
$$
(\Reeb_\L)_\alpha=\frac{\partial}{\partial s^\alpha}-W^{ji}_{\gamma\beta}\frac{\partial^2\L}{\partial s^\alpha\partial v^j_\gamma}\,\frac{\partial}{\partial v^i_\beta} \ .
$$

\subsection{The $k$-contact Euler--Lagrange equations}

As a result of the preceding definitions and results,
every \emph{regular} contact Lagrangian system 
has associated the $k$-contact Hamiltonian system
$(\oplus^k\Tan Q\times\R, \eta_\L^\alpha, E_\L)$. Then:

\begin{dfn}
Let $(\oplus^k \Tan Q\times\R^k,\L)$ be a $k$-contact Lagrangian system.

\noindent The \textbf{$k$-contact Euler--Lagrange equations} for a holonomic maps 
$\mbox{\boldmath $\sigma$}\colon\R^k\to\oplus^k \Tan Q\times\R^k$ are
\begin{equation}
\label{ELkcontact}
 \begin{cases}
\inn(\mbox{\boldmath $\sigma$}'_\alpha)\d\eta_\L^\alpha=
 \Big(\d E_\L-(\Lie_{(\Reeb_\L)_\alpha}E_\L)\eta_\L^\alpha\Big)\circ\mbox{\boldmath $\sigma$} \ ,\\
\inn(\mbox{\boldmath $\sigma$}'_\alpha)\eta_\L^\alpha=
-E_\L\circ\mbox{\boldmath$\sigma$} \ .
\end{cases}
\end{equation}
The \textbf{$k$-contact Lagrangian equations} for a $k$-vector field
${\bf X}_\L=((X_\L)_\alpha)$ in $\oplus^k \Tan Q\times\R^k$ are
\beq \begin{cases}
\inn((X_\L)_\alpha)\d\eta_\L^\alpha=\d E_\L- (\Lie_{(\Reeb_\L)_\alpha}E_\L)\eta_\L^\alpha \ ,\\
\inn((X_\L)_\alpha)\eta_\L^\alpha=-E_\L \ .
    \end{cases}
    \label{fieldLcontact}
\eeq
A $k$-vector field which is solution to these equations is called a
\textbf{Lagrangian $k$-vector field}.
\end{dfn}

A first relevant result is:

\begin{prop}
Let $(\oplus^k \Tan Q\times\R^k,\L)$ be a $k$-contact regular Lagrangian system.
Then, the $k$-contact Euler--Lagrange equations \eqref{fieldLcontact} admit solutions.
They are not unique if $k>1$.
\end{prop}
\begin{proof}
The proof is the same as that of Proposition 4.3 in \cite{GGMRR-2020}.
\end{proof}

In a natural chart of coordinates of $\oplus^k \Tan Q\times\R^k$, equations \eqref{ELkcontact} read
\beq
\label{ELeqs1}
\frac{\partial}{\partial t^\alpha}
\left(\frac{\displaystyle\partial \L}{\partial
v^i_\alpha}\circ{\mbox{\boldmath $\sigma$}}\right)=
\left(\frac{\partial \L}{\partial q^i}+
\displaystyle\frac{\partial\L}{\partial s^\alpha}\displaystyle\frac{\partial\L}{\partial v^i_\alpha}\right)\circ{\mbox{\boldmath $\sigma$}}
 \quad  , \quad
\derpar{s^\alpha}{t^\alpha}=\L\circ{\mbox{\boldmath $\sigma$}} \ ,
\eeq
meanwhile, for a $k$-vector field ${\bf X}_\L=((X_\L)_\alpha)$ with
$\displaystyle (X_\L)_\alpha= 
(X_\L)_\alpha^i\frac{\displaystyle\partial}{\displaystyle
\partial q^i}+(X_\L)_{\alpha\beta}^i\frac{\displaystyle\partial}{\displaystyle\partial v^i_\beta}+(X_\L)_\alpha^\beta\,\frac{\partial}{\partial s^\beta}$,
the Lagrangian equations \eqref{fieldLcontact} are
\bea
0 &=&
\displaystyle \Big((X_\L)_\alpha^j-v^j_\alpha\Big)
\frac{\partial^2\L}{\partial v^j_\alpha\partial s^\beta} \ ,
\label{A-E-L-eqs2}
\\
0 &=&
\displaystyle \Big((X_\L)_\alpha^j-v^j_\alpha\Big)
\frac{\partial^2\L}{\partial v^i_\beta\partial v^j_\alpha}
\label{A-E-L-eqs1} \ ,
\\
0 &=&
\displaystyle
\Big((X_\L)_\alpha^j-v^j_\alpha\Big)
\frac{\partial^2\L}{\partial q^i\partial v^j_\alpha}
+\frac{\partial\L}{\partial q^i}
-\frac{\partial^2\L}{\partial s^\beta\partial v^i_\alpha}(X_\L)_\alpha^\beta
\nonumber
\\ & &
-\frac{\partial^2\L}{\partial q^j \partial v^i_\alpha}(X_\L)_\alpha^j
-\frac{\partial^2\L}{\partial v^j_\beta\partial v^i_\alpha}(X_\L)_{\alpha\beta}^j
+\frac{\partial\L}{\partial s^\alpha}
\frac{\partial\L}{\partial v^i_\alpha}\ ,
\label{A-E-L-eqs3}
\\
0 &=&
\displaystyle \L + 
\frac{\partial\L}{\partial v^i_\alpha}\Big((X_\L)_\alpha^j-v^j_\alpha\Big)-(X_\L)_\alpha^\alpha\ .
\label{A-E-L-eqs4}
\eea
If $\L$ is a regular Lagrangian, equations \eqref{A-E-L-eqs1}
lead to $v^i_\alpha=(X_\L)_\alpha^i$, which
are the {\sc sopde} condition for the $k$-vector field ${\bf X}$.
Then, \eqref{A-E-L-eqs2} holds identically, and
\eqref{A-E-L-eqs4} and \eqref{A-E-L-eqs3} give 
\beann
(X_\L)_\alpha^\alpha&=& \L\ ,
\\
\displaystyle
-\frac{\partial\L}{\partial q^i}
+\frac{\partial^2\L}{\partial s^\beta\partial v^i_\alpha}(X_\L)_\alpha^\beta
+\frac{\partial^2\L}{\partial q^j \partial v^i_\alpha}v_\alpha^j
+\frac{\partial^2\L}{\partial v^j_\beta\partial v^i_\alpha}(X_\L)_{\alpha\beta}^j
&=&\frac{\partial\L}{\partial s^\alpha}
\frac{\partial\L}{\partial v^i_\alpha}\ .
\eeann
Notice that, if this {\sc sopde} ${\bf X}_\L$ is integrable,
these last equations are the Euler--Lagrange equations
\eqref{ELeqs1} for its integral maps.
In this way, we have proved that:

\begin{prop}
If $\L$ is a regular Lagrangian, then the corresponding
Lagrangian $k$-vector fields ${\bf X}_\L$
(solutions to the $k$-contact Lagrangian equations \eqref{fieldLcontact}) are 
{\sc sopde}'s and if, in addition, ${\bf X}_\L$ is integrable, then
its integral maps are solutions to the $k$-contact Euler--Lagrange field equations \eqref{ELkcontact}.

This {\sc sopde} ${\bf X}_\L\equiv\mbox{\boldmath $\Gamma$}_\L$ is called the \textbf{Euler--Lagrange $k$-vector field} 
associated with the Lagrangian function $\L$.
\end{prop}

\begin{obs}{\rm
It is interesting to point out how, in the Lagrangian formalism of dissipative field theories,
the second equation in \eqref{ELeqs1} relates the variation 
of the ``dissipation coordinates'' $s^\alpha$ to the Lagrangian function.
}
\end{obs}

\begin{obs}{\rm
If $\L$ is not regular then $(\oplus^k \Tan Q\times\R^k,\eta_\L^\alpha,E_\L)$
is a $k$-precontact system and, in general,
equations \eqref{ELkcontact} and \eqref{fieldLcontact} have no
solutions everywhere in $\oplus^k \Tan Q\times\R^k$ but, in the most favourable situations,
they do in a submanifold of $\oplus^k \Tan Q\times\R^k$ which is obtained by applying a suitable constraint algorithm.
Nevertheless, solutions to equations \eqref{fieldLcontact}
are not necessarily {\sc sopde}
(unless it is required as the additional condition
$J^\alpha(X_\alpha)=\Delta$) and,
as a consequence, if they are integrable, 
their integral maps are not necessarily holonomic.}
\end{obs}

\begin{obs}{\rm
Observe that the particular case $k=1$ gives the Lagrangian formalism
for mechanical systems with dissipation \cite{DeLeon2019,GGMRR-2019b}.
}\end{obs}

\subsection{\texorpdfstring{$k$}--contact canonical Hamiltonian formalism}

In the regular or the hyper-regular cases we have that ${\cal FL}$ is a (local) diffeomorphism between 
$(\oplus^k\Tan Q\times\R^k,\eta_\L^\alpha)$ and 
$(\oplus^k\Tan^*Q\times\R^k,\eta^\alpha)$,
where ${\cal FL}^{\;*}\eta^\alpha=\eta_\L^\alpha$.
Furthermore, there exists (maybe locally) a function 
$\H\in\Cinfty(\oplus^k\Tan^* Q\times\R)$ 
such that $\H=E_\L\circ {\cal F}\L^{-1}$; then we have the
$k$-contact Hamiltonian system $(\oplus^k\Tan^*Q\times\R^k,\eta^\alpha,\H)$,
for which ${\cal FL}_*({\Reeb}_\L)_\alpha={\Reeb}_\alpha$.
Therefore, if $\mbox{\boldmath $\Gamma$}_\L$ is an Euler--Lagrange $k$-vector field
associated with $\L$ in $\oplus^k\Tan Q\times\R^k$, then 
${\cal FL}_*\mbox{\boldmath $\Gamma$}_\L={\bf X}_\H$ is a
contact Hamiltonian $k$-vector field associated with $\H$ in
$\oplus^k\Tan^*Q\times\R^k$, and conversely.

For singular Lagrangians, following \cite{got79} we define:

\begin{dfn}
A singular Lagrangian $\L$ is \textbf{almost-regular} if
\begin{enumerate}
\item
$\mathcal{P}:= {\cal F}\L(\oplus^k \Tan Q\times\R^k)$
is a closed submanifold of $\oplus^k\Tan^*Q\times\R^k$.
\item
${\cal F}\L$ is a submersion onto its image.
\item
The fibres ${\cal F}\L^{-1}(p)$, for every $p \in \mathcal{P}$,
are connected submanifolds of $\oplus^k \Tan Q\times\R^k$.
\end{enumerate}
\end{dfn}

If $\L$ is almost-regular and $\jmath_0\colon{\cal P} \hookrightarrow\oplus^k\Tan^*Q\times\R^k$ is the natural embedding,
denoting by ${\cal F}\L_0\colon\oplus^k \Tan Q\times\R^k\to{\cal P}$ the restriction of $F\L$
given by $\jmath_0\circ{\cal F}\L_0={\cal F}\L$;
then there exists $\H_0\in\Cinfty({\cal P})$ such that
$({\cal F}\L_0)^*\H_0=E_\L$.
Furthermore, we can define $\eta^\alpha_0=\jmath_0^*\eta^\alpha$,
and then, the triple $({\cal P},\eta^\alpha_0,\H_0)$ is the
{\sl $k$-precontact Hamiltonian system associated with~$\L$},
and the corresponding Hamiltonian fields equations are
\eqref{hamilton-eqs-no-reeb} or \eqref{hec2} (in ${\cal P}$).
In general, these equations have no
solutions everywhere in ${\cal P}$ but, in the most favourable situations,
they do in a submanifold $P_f\hookrightarrow{\cal P}$, 
which is obtained applying a suitable constraint algorithm, 
and where there are Hamiltonian $k$-vector fields in ${\cal P}$, 
tangent to $P_f$.

\section{Symmetries and dissipated quantities in the Lagrangian formalism}
 \label{symms} 

As in \cite{GGMRR-2020},
we introduce different concepts of symmetry of the system, 
depending on which structure is preserved,
putting the emphasis on the transformations that leave
the geometric structures invariant,
or on the transformations that preserve the solutions of the system
(see, for instance \cite{Gracia2002,Roman2007}).
In this way, the following definitions and properties are adapted from those
stated for generic $k$-contact Hamiltonian systems to the case of
a $k$-contact regular Lagrangian system $(\oplus^k \Tan Q\times\R^k,\L)$;
that is, for the system $(\oplus^k \Tan Q\times\R^k,\eta_\L^\alpha,E_\L)$.
The proofs of the results for the general case are given in \cite{GGMRR-2020}.

\subsection{Symmetries}

\begin{dfn}
\label{def-dynsym}
Let $(\oplus^k \Tan Q\times\R^k,\L)$ be a $k$-contact regular Lagrangian system.
\begin{itemize}
\item 
A \textbf{Lagrangian dynamical symmetry} is a diffeomorphism 
$\Phi\colon\oplus^k \Tan Q\times\R^k \rightarrow\oplus^k \Tan Q\times\R^k$ such that, 
for every solution \mbox{\boldmath $\sigma$} to the $k$-contact Euler--Lagrange equations \eqref{ELkcontact}, 
$\Phi\circ\mbox{\boldmath $\sigma$}$ is also a solution.
\item 
An \textbf{infinitesimal Lagrangian dynamical symmetry} is a vector field 
$Y\in\mathfrak{X}(\oplus^k \Tan Q\times\R^k)$ whose local flow is made of local symmetries.
\end{itemize}
\end{dfn}

The following results give characterizations of symmetries in terms of $k$-vector fields:

\begin{lem}
Let $\Phi\colon\oplus^k \Tan Q\times\R^k\to\oplus^k \Tan Q\times\R^k$ 
be a diffeomorphism and $\mathbf{X}=(X_1,\dots,X_k)$ 
a $k$-vector field in $\oplus^k \Tan Q\times\R^k$. If $\psi$ is an
integral map of $\mathbf{X}$, then $\Phi\circ\psi$ is an integral map of 
$\Phi_*\mathbf{X}=(\Phi_*X_\alpha)$.
In particular, if ${\bf X}$ is integrable then $\Phi_*{\bf X}$ is also integrable.
\end{lem}

\begin{prop}
\label{prop-sym1}
If $\Phi \colon\oplus^k \Tan Q\times\R^k \rightarrow \oplus^k \Tan Q\times\R^k$ 
is a Lagrangian dynamical symmetry then, for every integrable $k$-vector field ${\bf X}$ solution to 
the $k$-contact Lagrangian equations \eqref{fieldLcontact}, 
$\Phi_*{\bf X}$ is another solution.

On the other side, if $\Phi$ transforms every $k$-vector field ${\bf X}_\L$ 
solution to the $k$-contact Lagrangian equations \eqref{fieldLcontact} into another solution, 
then for every integral map $\psi$ of ${\bf X}_\L$, we have that 
$\Phi\circ\psi$ is
a solution to the $k$-contact Euler--Lagrange equations \eqref{ELkcontact}.
\end{prop}

Among the most relevant symmetries are those that leave the geometric structures invariant:

\begin{dfn}
\label{def-hamconsym}
A \textbf{Lagrangian $k$-contact symmetry} is a diffeomorphism 
$\Phi\colon\oplus^k \Tan Q\times\R^k\rightarrow\oplus^k \Tan Q\times\R^k$ such that
    $$
        \Phi^*\eta_\L^\alpha=\eta_\L^\alpha
        \quad ,\quad \Phi^*E_\L=E_\L\ .
    $$
An \textbf{infinitesimal Lagrangian $k$-contact symmetry} is a vector field 
$Y\in \X(\oplus^k \Tan Q\times\R^k)$ whose local flow is a Lagrangian $k$-contact symmetry; that is,
    $$
        \Lie(Y)\eta_\L^\alpha=0
        \quad ,\quad \Lie(Y)E_\L=0 \ .
    $$
\end{dfn}

\begin{prop}
\label{prop-sym2}
Every (infinitesimal) Lagrangian $k$-contact symmetry preserves the Reeb vector fields, that is; 
$\Phi_*(\Reeb_\L)_\alpha=(\Reeb_\L)_\alpha$ (or $[Y,(\Reeb_\L)_\alpha]=0$).
\end{prop}

And, as a consequence of these results, we obtain the relation between 
these kinds of symmetries:

\begin{prop}
\label{prop-sym3}
(Infinitesimal) Lagrangian $k$-contact symmetries are (infinitesimal) 
Lagrangian dynamical symmetries.
\end{prop}

\subsection{Dissipation laws}

\begin{dfn}
\label{def-disip}
A map $F\colon M\rightarrow\R^k$, $F=(F^1,\dots,F^k)$, is said to satisfy:
\ben
\item
The \textbf{dissipation law for maps} if, for every map $\mbox{\boldmath $\sigma$}$ solution to the 
$k$-contact Euler--Lagrange equations \eqref{ELkcontact}, the divergence of 
$
F \circ\mbox{\boldmath $\sigma$}= (F^\alpha\circ\mbox{\boldmath $\sigma$}) \colon \R^k \to \R^k
$,
which is defined as usual by
$\ds
\mathrm{div} (F \circ\mbox{\boldmath $\sigma$})= 
{\partial (F^\alpha {\circ}\,\mbox{\boldmath $\sigma$})}/{\partial t^\alpha}
$,
satisfies that
\begin{equation}
\label{cons-law}
\mathrm{div} (F \circ\mbox{\boldmath $\sigma$})=
- \left[\strut (\Lie_{(\Reeb_\L)_\alpha}E_\L) F^\alpha \right] \circ\mbox{\boldmath $\sigma$}
\ .
\end{equation}
    \item 
The \textbf{dissipation law for $k$-vector fields} if, for every $k$-vector field ${\bf X}_\L$  
solution to the $k$-contact Lagrangian equations \eqref{fieldLcontact},
the following equation holds:
\beq
\label{cons-law field}
\Lie_{(X_\L)_\alpha}F^\alpha=-(\Lie_{(\Reeb_\L)_\alpha}E_\L)F^\alpha\ .
\eeq
\end{enumerate}
\end{dfn}

Both concepts are partially related by the following property:

\begin{prop}
\label{prop-disip}
If $F=(F^\alpha)$ satisfies the dissipation law for maps then,  
for every integrable $k$-vector field ${\bf X}_\L=((X_\L)_\alpha)$
which is a solution to the $k$-contact Lagrangian equations \eqref{fieldLcontact}, 
we have that the equation
\eqref{cons-law field} holds for ${\bf X}_\L$.

On the other side, if \eqref{cons-law field} holds for a $k$-vector field ${\bf X}$, then
\eqref{cons-law} holds for every integral map $\psi$ of ${\bf X}$.
\end{prop}

\begin{prop}
\label{braket}
If $Y$ is an infinitesimal dynamical symmetry then, for every solution ${\bf X}_\L=((X_\L)_\alpha)$
to the $k$-contact Lagrangian equations \eqref{fieldLcontact}, we have that
$$
\inn([Y,(X_\L)_\alpha])\eta_\L^\alpha=0 \quad , \quad
\inn([Y,(X_\L)_\alpha])\d\eta_\L^\alpha=0 \ .
$$
\end{prop}

Finally, we have the following fundamental result which associates dissipated quantities with symmetries:

\begin{thm} 
{\rm (Dissipation theorem)}. 
If $Y$ is an infinitesimal dynamical symmetry, 
then $F^\alpha=-i(Y)\eta_\L^\alpha$ 
satisfies the dissipation law for $k$-vector fields \eqref{cons-law field}.
\end{thm}

\subsection{Symmetries of the Lagrangian function}

Consider a $k$-contact regular Lagrangian system $(\oplus^k \Tan Q\times\R^k,\L)$.

First, remember that, if $\varphi\colon Q\to Q$ is a diffeomorphism,
we can construct the diffeomorphism 
$\Phi:=(\Tan^k\varphi, {\rm {\rm Id}_{\R^k}})\colon\oplus^k \Tan Q\times\R^k\longrightarrow 
\oplus^k \Tan Q\times\R^k$,
where $\Tan^k\varphi\colon\oplus^k \Tan Q\to\oplus^k \Tan Q$ denotes
the canonical lifting of $\varphi$ to $\oplus^k \Tan Q$.
Then $\Phi$ is said to be the {\sl canonical lifting} of $\varphi$ to $\oplus^k \Tan Q\times\R^k$.
Any transformation $\Phi$ of this kind
is called a \emph{natural transformation} of $\oplus^k \Tan Q\times\R^k$.

Moreover, given a vector field $Z\in \X(\oplus^k \Tan Q\times\R^k)$
we can define its {\sl complete lifting} to $\oplus^k \Tan Q\times\R^k$ as the vector field
$Y\in\X(\oplus^k \Tan Q\times\R^k)$ whose local flow is the canonical lifting of 
the local flow of $Z$ to $\oplus^k \Tan Q\times\R^k$;  that is, the vector field $Y=Z^C$,
where $Z^C$ denotes the complete lifting of $Z$ to $\oplus^k \Tan Q$,
identified in a natural way as a vector field in $\oplus^k \Tan Q\times\R^k$.
Any infinitesimal transformation $Y$ of this kind is called a \emph{natural infinitesimal transformation} 
of $\oplus^k \Tan Q\times\R^k$.

It is well-known that the canonical $k$-tangent structure $(J^\alpha)$ 
and the Liouville vector field $\Delta$ in $\oplus^k \Tan Q$ are invariant under the action
of canonical liftings of diffeomorphisms and vector fields from $Q$ to $\oplus^k \Tan Q$.
Then, taking into account the definitions of the canonical $k$-tangent structure $(J^\alpha)$ 
and the Liouville vector field $\Delta$ in $\oplus^k \Tan Q$,
it can be proved that canonical liftings of diffeomorphisms and vector fields
from $Q$ to $\oplus^k \Tan Q$ preserve these canonical structures as well 
as the Reeb vector fields $(\Reeb_\L)_\alpha$.

Therefore, as an immediate consequence, we obtain a relationship between 
Lagrangian-preserving natural transformations and contact symmetries:

\begin{prop}
If $\Phi\in{\rm Diff}(\oplus^k \Tan Q)$ (resp. $Y\in\vf(\oplus^k \Tan Q)$) is a canonical lifting
to $\oplus^k \Tan Q$ of a diffeomorphism $\varphi\in{\rm Diff}(Q)$
(resp.\ of a vector field $Z\in\vf(Q)$) that leaves the Lagrangian $\L$ invariant, then
it is a (infinitesimal) contact symmetry, {\it i.e.},
$$
    \Phi^*\eta_\L^\alpha=\eta_\L^\alpha \:,\
    \Phi^*E_\L=E_\L
    \qquad 
    ({\rm resp.}\ \Lie_Y\eta_\L^\alpha=0 \:,\
    \Lie_YE_\L=0 \:)\:.
$$
As a consequence, it is a (infinitesimal) Lagrangian dynamical symmetry.
\end{prop}

As an immediate consequence we have the following {\sl momentum dissipation theorem}:

\begin{prop}
If $\displaystyle\frac{\partial \L}{\partial q^i}=0$, 
then $\displaystyle\frac{\partial}{\partial q^i}$ 
is an infinitesimal contact symmetry and its associated dissipation law is given by the ``momenta''
$\displaystyle\left(\frac{\partial \L}{\partial v^i_\alpha}\right)$;
that is,  for every $k$-vector field ${\bf X}_\L=((X_\L)_\alpha)$  
solution to the $k$-contact Lagrangian  equations \eqref{fieldLcontact}, then
$$
\Lie_{(X_\L)_\alpha}\left(\frac{\partial \L}{\partial v^i_\alpha}\right)=
-(\Lie_{(\Reeb_\L)_\alpha}E_\L)\frac{\partial \L}{\partial v^i_\alpha}=
\derpar{\L}{s^\alpha}\frac{\partial \L}{\partial v^i_\alpha}\ .
$$
\end{prop}

\section{Examples}
\label{examples}

\subsection{An inverse problem for a class of elliptic and hyperbolic equations}

A generic second-order linear PDE in $\R^2$ is
$$
A u_{xx}+2Bu_{xy}+Cu_{yy}+Du_x+Eu_y+Fu+G=0 \, ,
$$
where $A,B,C,D,E,F,G$ are functions of $(x,y)$, with $A>0$. If $B^2-AC>0$
the equation is said to be hyperbolic, if $B^2-AC<0$ is
elliptic, and if $B^2-AC=0$ is parabolic. In $\R^n$ we consider the equation
\begin{equation}\label{eq:elliptic2}
A^{\alpha\beta}u_{\alpha\beta}+D^\alpha u_\alpha+G(u)=0 \, ,
\end{equation}
 where $1\leq \alpha,\beta\leq n$;
 and now we consider the following case: $A^{\alpha\beta}$ is constant and invertible (not parabolic), $D^\alpha$ is constant and $G$ is an arbitrary function in $u$.

In order to find a Lagrangian k-contact formulation of these kind of PDE's, consider $\oplus^nT\mathbb{R}\times \mathbb{R}^n$, with coordinates $(u,u_\alpha,s^\alpha)$ and a generic Lagrangian of the form
$$
L=\frac12a^{\alpha\beta}(u)u_\alpha u_\beta+b(u)u_\alpha s^\alpha+d(u,s)\, .
$$
The associated k-contact structure is given by
$$
\eta^\alpha=\d s^\alpha-\frac{\partial L}{\partial u_\alpha}\d u=\d s^\alpha-(a^{\alpha\beta}u_\beta+bs^\alpha+c^\alpha)\d u \, .
$$
The k-contact Euler--Lagrange equations associated to $L$ are
\begin{equation}
a^{\alpha\beta}u_{\alpha\beta}+\left(\frac12\frac{\partial a^{\alpha\beta}}{\partial u}-\frac12ba^{\alpha\beta}\right)u_{\alpha}u_{\beta}-\frac{\partial d}{\partial s^\beta}a^{\beta\alpha}u_\alpha+\left(-\frac{\partial d}{\partial s^\alpha}bs^\alpha+bd-\frac{\partial d}{\partial u}\right)=0 \, .
\end{equation}
If this equation has to match \eqref{eq:elliptic2} then
$$
a^{\alpha\beta}=A^{\alpha\beta}\,,\quad b=0\,,\quad d=-(a^{-1})_{\alpha\beta}D^\beta s^\alpha-\overline{g} ,
$$
 where $a=(a^{\alpha\beta})$
 and $\displaystyle\frac{\partial \overline{g}}{\partial u}=G$.

\paragraph{Damped vibrating membrane}

As a particular example consider the damped vibrating membrane,
which is described by the PDE
$$
u_{tt}-\mu^2(u_{xx}+u_{yy})+\gamma u_t=0 \, ;
$$
then
$$
A^{\alpha\beta}=\begin{pmatrix}
    1       & 0 & 0 \\
    0& -\mu^2 & 0 \\
    0&0&-\mu^2
    \end{pmatrix}\,,\quad
D^{\alpha}=\begin{pmatrix}
    \gamma    \\
     0 \\
    0
    \end{pmatrix}\,,\quad G=0 \, ,
$$
and therefore
$$
a^{\alpha\beta}=\begin{pmatrix}
    1       & 0 & 0 \\
    0& -\mu^2 & 0 \\
    0&0&-\mu^2
    \end{pmatrix}\,,\quad b=0\,,\quad
d=-\gamma s^t \, .
$$
 Then, a Lagrangian that leads to this equation is
$$
L=\frac12u_t^2-\frac{\mu^2}{2}(u_x^2+u_y^2)-\gamma s^t,
$$
for which
$$
\eta^t=\d s^t-u_t\d u \ , \ 
\eta^x=\d s^x+\mu^2u_x\d u \ , \ 
\eta^y=\d s^y+\mu^2u_y\d u \ .
$$
In this case, we have the contact symmetry $\displaystyle\frac{\partial }{\partial u}$
and the associated map ${\bf F}=(F^t,F^x,F^y)$ that satisfies the dissipation law for $3$-vector fields is
$$
F^t=-\inn(Y)\eta^t=u_t \ , \
F^x=-\inn(Y)\eta^x=-\mu^2u_x\ , \
F^y=-\inn(Y)\eta^y=-\mu^2u_y\ .
$$

\subsection{A vibrating string: Lorentz-like forces versus dissipation forces}

Terms linear in velocities can be found in Euler--Lagrange equations of symplectic systems. 
However, they have a specific form, arising from the coefficients of a closed $2$-form in the configuration space. 
The canonical example is the force of a magnetic field acting on a moving charged particle;
such forces do not dissipate energy. 
By contrast,
other forces linear in the velocities do dissipate energy;
for instance, damping forces.
To illustrate the difference between the equations arising from magnetic-like terms in the Lagrangian and the equations given by the $k$-contact formulation of a linear dissipation, we analyze the following academic example.

Consider an infinite string aligned with the $z$-axis, 
each of whose points can vibrate in a horizontal plane.
So, the independent variables are
$(t,z) \in \R^2$,
and the phase space is the bundle manifold 
$\oplus^2 \Tan \mathbb{R}^2$ 
with coordinates  $(x,y,x_t,x_z,y_t,y_z)$.
Let's imagine that the string is non-conducting,
but charged with linear density charge $\lambda$. 
Then, inspired by the Lagrangian formulation of the Lorentz force,
we set the Lagrangian 
$$
L_o=\frac12\rho(x_t^2+y_t^2)-\frac12\tau(x_z^2+y_z^2)-\lambda \left(\phi-A_1x_t-A_2y_t\right)
$$
depending on some fixed 
functions $A_1(x,y)$, $A_2(x,y)$ and $\phi(x,y)$. 
The resulting Euler--Lagrange equations are
\beq
\begin{aligned}
\rho x_{tt}-\tau x_{zz}&=-\lambda
\left(\frac{\partial A_2}{\partial x}- \frac{\partial A_1}{\partial y}\right)y_t+\lambda\frac{\partial \phi}{\partial x}\,,
\\
\rho y_{tt}-\tau y_{zz}&=\lambda \left(\frac{\partial A_2}{\partial x}- \frac{\partial A_1}{\partial y}
\right)x_t+\lambda\frac{\partial \phi}{\partial y}\,.
\end{aligned}
\label{equno}
\eeq
The left-hand side is the string equation with two modes of vibration in the plane $XY$ and in the right-hand side we have an electromagnetic-like term.

Now, consider the contact phase space 
$\oplus^2T\mathbb{R}^2\times \mathbb{R}^2$, 
with coordinates 
$(x,y,x_t,x_z,y_t,y_z,s^t,s^z)$. 
We add a simple dissipation term in the preceding Lagrangian:
$$
L = L_o+\gamma\,s^t =
\frac12 \rho(x_t^2+y_t^2) - \frac12 \tau(x_z^2+y_z^2) - \lambda \left(\phi-A_1x_t-A_2y_t\right) + \gamma s^t.
$$
The induced $2$-contact structure is
$$
\eta^t = 
\d s^t - 
(\rho x_t+\lambda A_1)\,\d x - (\rho y_t+\lambda A_2)\,\d y 
\,;
\quad 
\eta^z = 
\d s^z + \tau x_z \,\d x + \tau y_z\,\d y 
\,.
$$
The $2$-contact Euler--Lagrange equations are
\beq
\begin{aligned}
\rho x_{tt} - \tau x_{zz}
& = 
-\lambda \left(
\frac{\partial A_2}{\partial x} - \frac{\partial A_1}{\partial y}
\right)y_t 
+ \lambda\frac{\partial \phi}{\partial x} 
+ \gamma\rho x_t + \gamma\lambda A_1\, ,
\\
\rho y_{tt} - \tau y_{zz}
& =
\lambda \left(
\frac{\partial A_2}{\partial x}- \frac{\partial A_1}{\partial y}
\right)x_t
+ \lambda\frac{\partial \phi}{\partial y}
+ \gamma\rho y_t + \gamma\lambda A_2 \, .
\end{aligned}
\label{eqdos}
\eeq
Comparing equations \eqref{equno} and \eqref{eqdos} 
we observe that the dissipation originates two new terms: 
a dissipation force proportional to the velocity, 
and an extra term proportional to $(A_1,A_2)$. 
This last term comes from the non-linearity of 
the $2$-contact Euler--Lagrange equations 
with respect to the Lagrangian.

This system has the Lagrangian $2$-contact symmetry 
$$
Y = 
\frac{\partial A_2}{\partial x}\frac{\partial}{\partial x}
+ \frac{\partial A_1}{\partial y}\frac{\partial}{\partial y}\,.
$$
The associated map ${\bf F}=(F^t,F^z)$ that satisfies the dissipation law for $2$-vector fields is
\begin{eqnarray*}
F^t&=&
-\inn(Y)\eta^t =
\rho x_t\frac{\partial A_2}{\partial x}
+ \lambda \frac{\partial A_2}{\partial x}A_1
+ \rho y_t\frac{\partial A_1}{\partial y}
+ \lambda\frac{\partial A_1}{\partial y}A_2\ ,
\\
F^z&=&
-\inn(Y)\eta^z =
-\tau x_z\frac{\partial A_2}{\partial x}
-\tau y_z\frac{\partial A_1}{\partial y}\,.
\end{eqnarray*}

\section{Conclusions and outlook}

In a previous paper \cite{GGMRR-2020}
we introduced the notion of $k$-contact structure 
to describe Hamiltonian (De Donder--Weyl) covariant field theories with dissipation, 
bringing together contact Hamiltonian mechanics and $k$-symplectic field theory.

In this paper, we have developed the Lagrangian counterpart of this theory,
basing on contact Lagrangian and $k$-contact Hamiltonian formalisms. 
Thus, we have obtained and analyzed the Lagrangian (Euler--Lagrange) equations
of dissipative field theories. 
It should be pointed out that the regularity of the Lagrangian is required to obtain a $k$-contact structure.

We have also studied several kinds of symmetries: 
dynamical symmetries (those preserving solutions), 
$k$-contact symmetries (those preserving the $k$-contact structure and the energy)
and symmetries of the Lagragian function.
We have showed how to associate a dissipation law
with any dynamical symmetry. 

As interesting examples,
we have constructed 
Lagrangian functions for certain classes of elliptic
and hyperbolic partial differential equations;
in particular,
we have analyzed the example of the
damped vibrating membrane.
Another example has shown
the difference between the equations of 
the $k$-contact formulation of a linear dissipation 
and the equations arising from magnetic-like terms
appearing in some Lagrangian functions of field theories.

Among future lines of research,
the case of singular Lagrangians seems especially interesting,
though it would require to define the notions of \emph{$k$-precontact structure}
and \emph{$k$-precontact Hamiltonian system},
and to develop a constraint analysis to check the consistency of field equations.

\subsection*{Acknowledgments}

We acknowledge the financial support from the 
Spanish Ministerio de Ciencia, Innovaci\'on y Universidades project
PGC2018-098265-B-C33
and the Secretary of University and Research of the Ministry of Business and Knowledge of
the Catalan Government project
2017--SGR--932.

\bibliographystyle{abbrv}
\addcontentsline{toc}{section}{References}
\itemsep 0pt plus 1pt
\small

\begin{thebibliography}{10}

\bibitem{Awane1992}
A.~Awane.
\newblock k-symplectic structures.
\newblock {\em J. Math. Phys.} {\bf 33}(12):4046--4052, 1992.
(\url{https://doi.org/10.1063/1.529855}).

\bibitem{Banyaga2016}
A.~Banyaga and D.~F. Houenou.
\newblock {\em {A brief introduction to symplectic and contact manifolds}}.
\newblock World Scientific, Singapore, 2016.

\bibitem{Bravetti2017}
A.~Bravetti.
\newblock {Contact Hamiltonian dynamics: The concept and its use}.
\newblock {\em Entropy} {\bf 19}(10):535, 2017.
(\url{https://doi.org/10.3390/e19100535}).

\bibitem{Bravetti2018}
A.~Bravetti.
\newblock {Contact geometry and thermodynamics}.
\newblock {\em Int. J. Geom. Methods Mod. Phys.} {\bf 16}(supp01):1940003, 2018.
(\url{https://doi.org/10.1142/S0219887819400036}).

\bibitem{BCT-2017}
A.~Bravetti, H.~Cruz, and D.~Tapias.
\newblock {Contact Hamiltonian mechanics}.
\newblock {\em Ann. Phys. (N.Y.)} {\bf 376}:17--39, 2017.
(\url{https://doi.org/10.1016/j.aop.2016.11.003}).

\bibitem{Carinena1991}
J.~F. Cari{\~{n}}ena, M.~Crampin, and L.~A. Ibort.
\newblock {On the multisymplectic formalism for first order field theories}.
\newblock {\em Diff. Geom. Appl.} {\bf 1}(4):345--374, 1991.
(\url{https://doi.org/10.1016/0926-2245(91)90013-Y}).

\bibitem{CG-2019}
J.~Cari{\~{n}}ena and P.~Guha.
\newblock Nonstandard Hamiltonian structures of the Liénard equation
and contact geometry.
\newblock {\it Int. J. Geom. Meth. Mod. Phys.} {\bf 16}(supp 01), 1940001 (2019).
(\url{https://doi.org/10.1142/S0219887819400012}).

\bibitem{CCM-2018}
F.~M. Ciaglia, H.~Cruz, and G.~Marmo.
\newblock {Contact manifolds and dissipation, classical and quantum}.
\newblock {\em Ann. Phys. (N.Y.)} {\bf 398}:159--179, 2018.
(\url{https://doi.org/10.1016/j.aop.2018.09.012}).

\bibitem{DeLeon2019}
M.~de~Le{\'{o}}n and M.~Lainz-Valc{\'{a}}zar.
\newblock {Singular Lagrangians and precontact Hamiltonian systems}.
\newblock {\em Int. J. Geom. Meth. Mod.} 
{\bf 16}(19):1950158, 2019.
(\url{https://doi.org/10.1142/S0219887819501585}).

\bibitem{DeLeon1988}
M.~de~Le{\'{o}}n, I.~M\'{e}ndez, and M.~Salgado.
\newblock {p-Almost tangent structures}.
\newblock {\em Rend. Circ. Mat. Palermo} {\bf 37}(2):282--294, 1988.
(\url{https://doi.org/10.1007/BF02844526}).

\bibitem{DeLeon1988a}
M.~de~Le{\'{o}}n, I.~M\'{e}ndez, and M.~Salgado.
\newblock {Regular p-almost cotangent structures}.
\newblock {\em J. Korean Math. Soc.} {\bf 25}(2):273--287, 1988.

\bibitem{DeLeon2015}
M.~de~Le{\'{o}}n, M.~Salgado, and S.~Vilari{\~{n}}o.
\newblock {\em {Methods of Differential Geometry in Classical Field Theories}}.
\newblock World Scientific, 2015.

\bibitem{DeLeon2016b}
M.~de~Le{\'{o}}n and C.~Sard{\'{o}}n.
\newblock {Cosymplectic and contact structures to resolve time-dependent and
  dissipative Hamiltonian systems}.
\newblock {\em J. Phys. A: Math. Theor.} {\bf 50}(25):255205, 2017.
(\url{https://doi.org/10.1088/1751-8121/aa711d}).

\bibitem{EMR-96}
 A. Echeverr\'\i a-Enr\'\i quez, M.C. Mu\~noz-Lecanda, and N. Rom\'an-Roy.
 Geometry of Lagrangian first-order classical field theories,
\newblock {\em Fortschr. Phys.} {\bf 44}:235--280, 1996.
(\url{https://doi.org/10.1002/prop.2190440304}).

\bibitem{GGMRR-2020}
J.~{Gaset}, X.~{Gr\`acia}, M.~{Mu\~noz-Lecanda}, X.~{Rivas}, and
  N.~{Rom\'an-Roy}.
\newblock {A contact geometry framework for field theories with dissipation}.
\newblock {\em Annals of Physics} {\bf 414}, 168092, 2020.
(\url{https://doi.org/10.1016/j.aop.2020.168092}).

\bibitem{GGMRR-2019b}
J.~{Gaset}, X.~{Gr\`acia}, M.~{Mu\~noz-Lecanda}, X.~{Rivas}, and
  N.~{Rom\'an-Roy}.
\newblock {New contributions to the Hamiltonian and Lagrangian contact formalisms for dissipative mechanical systems and their symmetries}.
\newblock {\em arXiv preprint arXiv:1907.02947} (2019).

\bibitem{Geiges2008}
H.~Geiges.
\newblock {\em {An Introduction to Contact Topology}}.
\newblock Cambridge University Press, Cambridge, 2008.

 \bibitem{GIMMSY-mm}
M.J. Gotay, J. Isenberg, and J.E. Marsden.
Momentum maps and classical relativistic f\/ields. I.~Covariant theory,
\newblock {\em arXiv:physics/9801019v2}, 2004.

\bibitem{got79}
M.~J. Gotay and J.~M. Nester.
\newblock Presymplectic {L}agrangian systems {I}: the constraint algorithm and
  the equivalence theorem.
\newblock {\em Ann. Inst. Henri Poincaré} {\bf 30}(2):129--142, 1979.

\bibitem{Goto-2016}
S. Goto.
\newblock Contact geometric descriptions of vector fields on dually flat spaces
  and their applications in electric circuit models and nonequilibrium
  statistical mechanics.
\newblock {\em J. Math. Phys.} {\bf 57}(10):102702, 2016.
(\url{https://doi.org/10.1063/1.4964751}).

\bibitem{Gracia2000}
X.~Gr{\`{a}}cia.
\newblock {Fibre derivatives: some applications to singular {L}agrangians}.
\newblock {\em Rep. Math. Phys.} {\bf 45}(1):67--84, 2000.
(\url{https://doi.org/10.1016/S0034-4877(00)88872-2}).

\bibitem{Gracia2002}
X.~Gr{\`{a}}cia and J.~Pons.
\newblock {Symmetries and infinitesimal symmetries of singular differential
  equations}.
\newblock {\em J. Phys. A Math. Gen.} {\bf 35}(24):5059--5077, 2002.
(\url{https://doi.org/10.1088/0305-4470/35/24/306}).

\bibitem{GRR-2019}
X.~{Gr\`acia}, X.~{Rivas}, and  N.~{Rom\'an-Roy}.
\newblock {Constraint algorithm for singular field theories in the $k$-cosymplectic framework.}
\newblock {\em J. Geom. Mech.} {\bf 12}(1):1--23, 2020.
(\url{https://doi.org/10.3934/jgm.2020002}).

\bibitem{Kholodenko2013}
A.~L. Kholodenko.
\newblock {\em {Applications of Contact Geometry and Topology in Physics}}.
\newblock World Scientific, 2013.

\bibitem{Lainz2018}
M.~Lainz-Valc{\'{a}}zar and M.~de~Le{\'{o}}n.
\newblock {Contact Hamiltonian Systems}.
\newblock {\em J. Math. Phys.} {\bf 60}(10):102902 (2019).
(\url{https://doi.org/10.1063/1.5096475}).

\bibitem{Lee2013}
J.~M. Lee.
\newblock {\em {Introduction to Smooth Manifolds}}.
\newblock Springer, 2013.

\bibitem{LIU2018}
Q.~Liu, P.~J. Torres, and C.~Wang.
\newblock Contact Hamiltonian dynamics: Variational principles, invariants,
  completeness and periodic behavior.
\newblock {\em Ann. Phys.} {\bf 395}:26 -- 44, 2018.
(\url{https://doi.org/10.1016/j.aop.2018.04.035}).

\bibitem{RMS-2017}
H.~Ram\'irez, B.~Maschke, and D.~Sbarbaro.
\newblock Partial stabilization of input-output contact systems on a {L}egendre
  submanifold.
\newblock {\em IEEE Trans. Automat. Control} {\bf 62}(3):1431--1437, 2017.
(\url{https://doi.org/10.1109/TAC.2016.2572403}).

\bibitem{Rey2004}
A.~M. {Rey-Roca}, N.~Rom{\'{a}}n-Roy, and M.~Salgado.
\newblock {Gunther's formalism (k-symplectic formalism) in classical field theory: Skinner--Rusk approach and the evolution operator}.
\newblock {\em J. Math. Phys.} {\bf 46}(5):052901, 2005.
(\url{https://doi.org/10.1063/1.1876872}).

\bibitem{Roman2011}
A.~M. {Rey Roca}, N.~Rom{\'{a}}n-Roy, M.~Salgado, and S.~Vilari{\~{n}}o.
\newblock {On the k-Symplectic, k-Cosymplectic and Multisymplectic Formalisms of Classical Field Theories}.
\newblock {\em J. Geom. Mech.} {\bf 3}(1):113--137, 2011.
(\url{https://doi.org/10.3934/jgm.2011.3.113}).

\bibitem{Roman2009}
N.~Rom{\'{a}}n-Roy.
\newblock {Multisymplectic Lagrangian and Hamiltonian Formalisms of Classical
  Field Theories}.
\newblock {\em SIGMA Symmetry Integrability Geom. Methods Appl.} {\bf 5}:100, 2009.
(\url{https://doi.org/10.3842/SIGMA.2009.100}).

\bibitem{Roman2007}
N.~Rom{\'{a}}n-Roy, M.~Salgado, and S.~Vilari{\~{n}}o.
\newblock {Symmetries and Conservation Laws in the G{\"{u}}nther k-symplectic
  Formalism of Field Theory}.
\newblock {\em Rev. Math. Phys.} {\bf 19}(10):1117--1147, 2007.
(\url{https://doi.org/10.1142/S0129055X07003188}).

\end{thebibliography}

\end{document}